\documentclass{article}

\usepackage{graphicx} 
\usepackage[hidelinks]{hyperref} 
\usepackage{natbib} 
\usepackage{subcaption} 
\usepackage{enumitem} 
\usepackage{pifont} 
\usepackage{comment}

\usepackage[skip=10pt plus1pt, indent=0pt]{parskip} 
\usepackage[margin=1in]{geometry} 

\usepackage{amssymb,amsmath,amsthm,amsfonts}
\usepackage{mathtools}
\usepackage{bbm}
\usepackage{color,xcolor}

\usepackage{algorithm}
\usepackage{algorithmicx}
\usepackage{algpseudocode}

\usepackage{booktabs}

\newcommand{\one}{\mathbbm{1}}

\newcommand{\Q}{\operatorname{\mathsf{P}}}
\newcommand{\R}{\mathbb{R}}

\newcommand{\Xx}{\mathcal{X}}

\newcommand{\Pn}{\hat{\mathbb{P}}_n}
\newcommand{\PN}{\hat{\mathbb{P}}_N}
\newcommand{\Pnp}{\hat{\mathbb{P}}_{n+1}}

\newtheorem{propo}{Proposition}
\newtheorem{lemma}[propo]{Lemma}
\newtheorem{theorem}[propo]{Theorem}

\theoremstyle{definition}
\newtheorem{definition}[propo]{Definition}
\newtheorem{remark}[propo]{Remark}
\newtheorem{example}[propo]{Example}

\title{Residual Distribution Predictive Systems}
\date{}
\author{Sam Allen\thanks{Institute of Statistics, Karlsruhe Institute of Technology, Karlsruhe, Germany. \url{sam.allen@kit.edu}} \and Enrico Pescara\thanks{Seminar for Statistics, ETH Zurich, Zurich, Switzerland. \url{epescara@student.ethz.ch}} \and Johanna Ziegel\thanks{Seminar for Statistics, ETH Zurich, Zurich, Switzerland. \url{ziegel@stat.math.ethz.ch}} }

\begin{document}

\maketitle

\begin{abstract}
    Conformal predictive systems are sets of predictive distributions with theoretical out-of-sample calibration guarantees. The calibration guarantees are typically that the set of predictions contains a forecast distribution whose prediction intervals exhibit the correct marginal coverage at all levels. Conformal predictive systems are constructed using conformity measures that quantify how well possible outcomes conform with historical data. However, alternative methods have been proposed to construct predictive systems with more appealing theoretical properties. We study an approach to construct predictive systems that we term Residual Distribution Predictive Systems. In the split conformal setting, this approach nests conformal predictive systems with a popular class of conformity measures, providing an alternative perspective on the classical approach. In the full conformal setting, the two approaches differ, and the new approach has the advantage that it does not rely on a conformity measure satisfying fairly stringent requirements to ensure that the predictive system is well-defined; it can readily be implemented alongside any point-valued regression method to yield predictive systems with out-of-sample calibration guarantees. The empirical performance of this approach is assessed using simulated data, where it is found to perform competitively with conformal predictive systems. However, the new approach offers considerable scope for implementation with alternative regression methods.  
\end{abstract}

\section{Introduction}\label{sec:intro}

For forecasts to be useful for decision making, they must align statistically with what actually materialises. Such forecasts are said to be \emph{calibrated}. For example, if a binary event is predicted to occur with a certain probability, then it should manifest with this probability. Similarly, prediction intervals should contain the outcome with the desired coverage level, and full predictive distributions should yield calibrated prediction intervals. We focus here on probabilistic forecasts for real-valued outcomes, where the goal is to issue predictive distributions that are as informative as possible subject to being calibrated \citep{gneiting2007b}.

Forecasting is typically performed out-of-sample, with predictions issued for future outcomes that have not yet been observed. Ideally, a forecasting procedure would issue calibrated predictions for out-of-sample outcomes by construction, possibly under some weak assumptions on the data generating process. However, designing a forecasting procedure with such out-of-sample calibration guarantees is generally infeasible \citep{vovk2022book}. Alternatively, conformal prediction has recently become popular in the machine learning literature since it goes one step in this direction by issuing a set of probabilistic predictions (i.e. a credal set) that is guaranteed to contain an out-of-sample calibrated prediction, under the assumption of exchangeable data \citep{GammermanEtAl1998,ShaferVovk2008,vovk2017,FontanaEtAl2023}. If the size of the credal set is small, then any forecast within the set should be approximately calibrated. If the size of the credal set is large, then it suggests there is large epistemic uncertainty in the forecast \citep{GammermanVovk2002,HullermeierWaegeman2021}.

When issuing probabilistic forecasts for real-valued outcomes, the credal sets are generally referred to as \emph{predictive systems}. While they have received less attention in the literature than conformal prediction methods for classification and regression tasks, conformal predictive systems contain considerably more information; they directly provide prediction intervals with out-of-sample coverage guarantees at any coverage level; see for example \citet{vovk2017}.

The classical conformal prediction framework is based on conformity measures (or non-conformity scores), which measure how well possible covariate-observation pairs \emph{conform} with previously observed data. \cite{allen2025} demonstrate that this approach essentially corresponds to a particular example of a more general framework that constructs credal sets with out-of-sample calibration guarantees using forecasting procedures that are calibrated in-sample. \cite{van2025} show a similar result in the context of point forecasting. An obvious question is what other forecasting procedures are calibrated in-sample. \cite{allen2025} study in detail two approaches when interest is on probabilistic forecasts of real-valued outcomes: binning (or analogue) prediction methods, and isotonic distributional regression \citep[IDR;][]{HenziEtAl2021}. The latter approach provides a generalisation of the Venn-Abers predictors of \citet{VovkPetejETAL2015}.

In this work, we analyse a simple alternative forecasting procedure that is commonly employed in practice. The approach fits a regression model to some training data, and uses this to obtain a point forecast for the new outcome. This point forecast is then dressed with the residuals in a training data set to obtain a (discrete) forecast distribution. The resulting forecast distributions satisfy a popular notion of forecast calibration in-sample, and we therefore demonstrate how it can be used to obtain predictive systems that are guaranteed to contain a calibrated out-of-sample predictive distribution. We refer to these as \emph{Residual Distribution Predictive Systems}. 

This approach shares similarities with the classical conformal prediction framework, since conformity measures are typically constructed using regression residuals. We demonstrate that in a split conformal setting \citep{vovk2018b}, these two approaches actually coincide, offering an alternative perspective on the original conformal prediction system framework. In the full conformal setting, the two approaches differ, and while the classical approach requires that the conformity measure satisfies a fairly strict monotonicity condition \citep{vovk2017}, the novel approach does not. This facilitates the construction of conformal predictive systems using any regression method, without the need to check that the monotonicity condition is satisfied. However, in practice, some regression methods can yield predictive systems with excessively large thickness. In practice, we find that the two approaches yield predictive systems with similar predictive ability, though there is additional scope to apply the Residual Distribution Predictive System with more flexible regression models.

The following section defines predictive systems and notions of calibration, and gives the construction result for predictive systems with out-of-sample calibration guarantees. The classical conformal predictive systems are introduced in Section \ref{sec:cmps}, while the residual distribution predictive systems that we study are introduced in Section \ref{sec:rdps}. Some simulation results are provided in Section \ref{sec:simstudy}, before we conclude in Section \ref{sec:conc}. Code to reproduce the simulation results is available at \url{https://github.com/sallen12/RDPS_rep}.

\section{Predictive Systems}\label{sec:ps}

Suppose our goal is to predict an outcome $Y_{n+1} \in \R$ using a training data set $(x_1, y_1)$, \dots, $(x_n, y_n) \in \Xx \times \R$ and a new covariate $x_{n+1} \in \Xx$. A prediction for $Y_{n+1}$ is typically obtained using state-of-the-art statistical and machine learning models, which take the training data and the new covariate as inputs, and output a prediction for the unknown outcome. In the following, we assume that the predictions are probability distributions on $\R$, and we refer to such a model as a \emph{forecasting procedure}. The training data can be represented via the empirical distribution $\Pn = (1/n)\sum_{i=1}^n \delta_{(x_i, y_i)}$, and we denote the forecasting procedure corresponding to the training data $\Pn$ and covariate $x \in \Xx$ as the function $G_{x}^{\Pn} : \R \to [0, 1]$. In practice, the forecasting procedure is evaluated at $x = x_{n+1}$ to get a prediction for $Y_{n+1}$. Throughout, we use lower and upper case to denote deterministic and random elements, respectively, and we denote the underlying probability measure by $\Q$.

Ideally, the forecasting procedure should issue a forecast for $Y_{n+1}$ that is calibrated, in the sense that the forecast is statistically compatible with the corresponding outcome. There are several different ways to define forecast calibration, though it is most common in practice to assess \emph{probabilistic calibration}, which corresponds to the prediction intervals derived from the predictive distributions achieving the correct marginal coverage \citep{dawid1984,diebold1998,gneiting2007b}. Following \cite{allen2025}, we define probabilistic calibration as follows.
\\
\begin{definition}\label{def:probcal}
	A (random) forecast distribution $G$ is \emph{probabilistically calibrated} for $Y \in \R$ if
	\begin{equation}\label{eq:probcal}
	    \Q(G(Y) \le \alpha) \leq \alpha \leq \Q(G(Y-) < \alpha) \quad \text{for all } \alpha \in (0,1),
	\end{equation}
	where $G(y-) = \lim_{z \uparrow y} G(z)$.
\end{definition}
If the forecast distribution is continuous, probabilistic calibration implies that the forecast probability integral transform $G(Y)$ follows a standard uniform distribution. Stronger notions of calibration have also been proposed, which generally correspond to \eqref{eq:probcal} conditioned on more information \citep{GneitingRanjan2013,Tsyplakov2014,GneitingResin2023}, though probabilistic calibration remains a useful notion that can be leveraged in decision-making contexts to achieve asymptotically efficient decisions \citep[see][Section 7.7]{vovk2022book}. We can similarly define the calibration of a forecasting procedure. In this case, it is important to distinguish between \emph{in-sample} and \emph{out-of-sample} calibration.
\\
\begin{definition}\label{def:probcal_G}
    Let $(X_1, Y_1), \dots, (X_n, Y_n) \in \Xx \times \R$ be a random sample with empirical distribution $\Pn = (1/n)\sum_{i=1}^n \delta_{(X_i, Y_i)}$. A forecasting procedure $G$ is \emph{in-sample probabilistically calibrated} if
    $G^{\Pn}_{X_i}$ is a probabilistically calibrated forecast for $Y_i$ under $\Pn$, for all $i = 1, \dots, n$. That is,
    \[
    \Pn\big(G^{\Pn}_{X_i}(Y_i) \le \alpha\big) \le \alpha \le \Pn\big(G^{\Pn}_{X_i}(Y_i-) < \alpha\big) \quad \text{for all $\alpha \in (0,1)$ almost surely.}
    \]
\end{definition}

Analogously, a forecasting procedure $G$ could be defined as out-of-sample calibrated if, given a random covariate $X_{n+1}$, $G^{\Pn}_{X_{n+1}}$ is a probabilistically calibrated forecast for $Y_{n+1}$ under $\Q$. Since most forecasting tasks are out-of-sample, we generally desire forecasts that are out-of-sample calibrated. However, deriving a forecasting procedure that is guaranteed to issue out-of-sample calibrated forecasts is generally not possible \citep{vovk2022book}. Instead, conformal prediction has become popular since it provides a framework for obtaining credal sets that are guaranteed to contain an out-of-sample calibrated forecast, under the assumption of exchangeability. When dealing with probabilistic forecasts for real-valued outcomes, these credal sets are typically referred to as \emph{conformal predictive systems}. We define a predictive system as follows.
\\
\begin{definition}
	A \emph{predictive system} is a set $\Pi \subseteq \mathbb{R} \times [0,1]$ of the form
	\[
	\Pi = \{(y, \tau) \in \mathbb{R} \times [0,1] \mid \Pi_{\ell}(y) \leq \tau \leq \Pi_u(y)\},
	\]
	where the lower and upper bounds $\Pi_{\ell}, \Pi_u : \mathbb{R} \to [0,1]$ are increasing functions satisfying $\Pi_{\ell}(y) \leq \Pi_u(y)$ for all $y \in \mathbb{R}$, and
	\[
	\lim_{y \to -\infty} \Pi_{\ell}(y) = 0, \qquad \lim_{y \to \infty} \Pi_u(y) = 1.
	\]
	The \emph{thickness} of the predictive system $\Pi$ is defined as
	\[
	\operatorname{th}(\Pi) = \inf\{\varepsilon > 0 \mid \Pi_u(y) - \Pi_{\ell}(y) \leq \varepsilon \text{ for all but finitely many } y \in \mathbb{R}\}.
	\]
\end{definition}
A predictive system is therefore comprised of lower and upper bounds $\Pi_{\ell}, \Pi_u$, which can be interpreted as two stochastically ordered (possibly defective) distribution functions, with the region between the two bounds defining a set of probabilistic forecasts. A predictive distribution is said to be contained in a predictive system $\Pi$ if it lies between the system's lower and upper bounds, $\Pi_\ell$ and $\Pi_u$. Conformal predictive systems are constructed such that they are guaranteed to contain an out-of-sample probabilistically calibrated forecast distribution for $Y_{n+1}$. In practice, this guarantee is only useful if the difference between the bounds (the thickness) is not too large.
 
An obvious question is if and how we can construct predictive systems that are guaranteed to contain an out-of-sample calibrated forecast distribution. \cite{allen2025} describe how they can be obtained via a general framework that leverages in-sample calibrated forecasting procedures. In particular, given a procedure $G$, training data $(x_1, y_1), \dots, (x_n, y_n)$, and covariate $x_{n+1}$, a predictive system can be defined by the bounds
\begin{equation}\label{eq:Gbands}
    \Pi_{\ell,x_{n+1}}^{\Pn}(y) = \inf_{y' \in \mathbb{R}} G_{x_{n+1}}^{\Pnp(y')}(y), \qquad \Pi_{u,x_{n+1}}^{\Pn}(y) = \sup_{y' \in \mathbb{R}} G_{x_{n+1}}^{\Pnp(y')}(y), \quad y \in \R,
\end{equation}
where $\Pnp(y')$ is used to denote the empirical distribution of $(x_1, y_1), \dots, (x_n, y_n), (x_{n+1}, y')$, the training data along with the new covariate $x_{n+1}$ and an assumed value $y' \in \R$ of the new outcome. These bounds essentially correspond to the pointwise lowest and highest values that the predictive distribution can attain when the forecasting procedure is trained using the training data augmented by the new pair $(x_{n+1}, y')$, for all possible $y'$. If a large amount of training data is available, then this additional training data point will typically have less influence on the learned predictive distribution, resulting in a predictive system whose lower and upper bounds are close together (and vice versa). 

\citet[Theorem 1]{allen2025} demonstrate that if the forecasting procedure $G$ is in-sample calibrated, then the predictive system $\Pi$ generated by $G$ is guaranteed to contain an out-of-sample calibrated forecast for $Y_{n+1}$. 
\\
\begin{theorem}[\cite{allen2025}]\label{thm:intoout}
    Let $(X_1, Y_1), \dots, (X_{n+1}, Y_{n+1}) \in \Xx \times \R$ be exchangeable with empirical distribution $\Pn = (1/n)\sum_{i=1}^n \delta_{(X_i, Y_i)}$. If $G$ is an in-sample probabilistically calibrated forecasting procedure, then the predictive system defined by the bounds at \eqref{eq:Gbands} satisfies
	\begin{equation}\label{eq:OoS}
    \Q\left(\Pi_{\ell,X_{n+1}}^{\Pn}(Y_{n+1}) \le \alpha\right) \leq \alpha \leq
    \Q\left(\Pi_{u,X_{n+1}}^{\Pn}(Y_{n+1}-) < \alpha\right) \quad \text{for all $ \alpha \in (0,1) $}.
	\end{equation}
\end{theorem}
\bigskip
\begin{remark}
    Equation \eqref{eq:OoS} in Theorem \ref{thm:intoout} continues to hold almost surely when we replace the probability measure $\Q$ with the (random) empirical distribution $\Pnp$. This statement is stronger, and can be used to test the assumption of exchangeability \citep[Section 8]{vovk2022book}.
\end{remark}

The bounds at \eqref{eq:Gbands} therefore provide a general framework for constructing predictive systems based on existing, well-understood forecasting procedures. In a split conformal setting (see below), the classical conformal prediction framework corresponds to a particular choice of forecasting procedure $G$, provided in the following section. More generally, predictive systems with out-of-sample calibration guarantees can be obtained from any in-sample calibrated forecasting procedure. \cite{allen2025} use this to introduce predictive systems based on binning prediction methods and isotonic distributional regression \citep[IDR;][]{HenziEtAl2021}, for example. A further in-sample calibrated forecasting procedure is given in \citet[Example 1]{allen2025} but not studied in detail. This procedure is the main focus of this paper (see Section \ref{sec:rdps}).
\\
\begin{remark}[Split conformal prediction]
    In practice, calculating the lower and upper bounds of the predictive system at \eqref{eq:Gbands} can be prohibitively computationally expensive, since this may require refitting the forecasting procedure for every possible new outcome $y' \in \R$. Instead, if sufficient data is available, a more efficient approach is to divide the training data $(x_1, y_1), \dots, (x_n, y_n)$ into an \emph{estimation set} $(x_1, y_1), \dots, (x_N, y_N)$ and a \emph{calibration set} $(x_{N+1}, y_{N+1}), \dots, (x_n, y_n)$, for $N < n$, and to estimate some parameters of the forecasting procedure on the estimation set. It then often becomes clear for which $y'$ the bounds of the predictive system at \eqref{eq:Gbands} are attained, and they can be obtained easily using the calibration data. In this case, the model parameters would not need to be relearned for every possible $y' \in \R$, substantially simplifying the calculation of the predictive system. This approach is often referred to as a \emph{split conformal} framework, in contrast to the \emph{full conformal} framework previously described \citep{vovk2018b}. The results of Theorem \ref{thm:intoout} hold in both the split and conformal settings \citep{allen2025}, and, unless specified otherwise, we adopt the notation of the full conformal setting throughout.
\end{remark}

\section{Conformal Predictive Systems}\label{sec:cmps}

\cite{vovk2017} initially defined conformal predictive systems using conformity measures. A conformity measure is a real-valued function $A$ that quantifies how much a prospective covariate-outcome pair conforms with the training data. The conformity measure is often of the form 
\begin{equation}\label{eq:conf}
    A(\Pn, (x, y)) := y - \hat{y}_x,
\end{equation}
for $(x, y) \in \Xx \times \R$, where $\hat{y}_x \in \R$ is a prediction for $y$ computed from $x$ and the training data $(x_1, y_1), \dots, (x_n, y_n)$. The prediction could be obtained using ordinary least squares regression, for example, or more flexible alternatives. A predictive system corresponding to $A$ is defined by the bounds
\begin{equation}\label{eq:CM_low}
    \Pi_{\ell, x_{n+1}}^{\Pn}(y) = \frac{1}{n + 1} \sum_{i = 1}^n \one\{A(\Pnp(y),(x_i, y_i)) < A(\Pnp(y),(x_{n+1}, y))\},
\end{equation}
\begin{equation}\label{eq:CM_upp}
    \Pi_{u, x_{n+1}}^{\Pn}(y) = \frac{1}{n + 1} \left( 1 + \sum_{i = 1}^n \one\{A(\Pnp(y),(x_i, y_i)) \leq A(\Pnp(y),(x_{n+1}, y)) \} \right).
\end{equation}
We refer to these in the following as \emph{Conformal Predictive Systems}.

In the full conformal setting, it is sufficient to define the conformity measure $A(\Pn, (x, y))$ only at pairs $(x, y)$ that are in the support of $\Pn$ (see \eqref{eq:CM_low} and \eqref{eq:CM_upp}). However, calculating the lower and upper bounds of the predictive system is generally computationally expensive and often practically infeasible; for example, for a conformity measure of the form \eqref{eq:conf}, it requires fitting a regression model to $(x_1, y_1), \dots ,(x_n, y_n), (x_{n+1}, y)$ for all possible values of $y \in \R$. In the split conformal setting, $\Pnp(y)$ in \eqref{eq:CM_low} and \eqref{eq:CM_upp} is replaced with $\hat{\mathbb{P}}_{N}$, the empirical distribution of the data in the estimation set, $(x_1, y_1), \dots, (x_N, y_N)$. Since this does not depend on $y$, the lower and upper bounds of the prediction system are generally much easier to calculate. In this case, the conformal predictive system can also be expressed as the predictive system generated using the bounds at \eqref{eq:Gbands} with the forecasting procedure
\begin{equation}\label{eq:CM_G}
    G_x^{\Pn}(y) = \frac{1}{n - N} \sum_{i=N+1}^{n} \one\{A(\PN, (x_i, y_i)) \le A(\PN, (x, y))\},
\end{equation}
where $\PN$ is the empirical distribution of the estimation data; this empirical distribution $\PN$ does not change when $\Pn$ is replaced with $\Pnp(y)$ in \eqref{eq:Gbands}. This representation is useful to study the connection with the Residual Distribution Predictive Systems introduced in Section \ref{sec:rdps}.

The forecasting procedure at \eqref{eq:CM_G} outputs valid predictive distributions when $A(\PN,(x_{n+1},y))$ is increasing in $y$. In this case, the forecasting procedure is in-sample probabilistically calibrated. Hence, conformal predictive systems are guaranteed to contain an out-of-sample probabilistically calibrated forecast distribution in the split conformal framework. This is also true in the full conformal framework, but only when a stronger constraint on the conformity measure $A$ is satisfied. Importantly, for certain choices of $A$, the lower and upper bounds at \eqref{eq:CM_low} and \eqref{eq:CM_upp} may not be increasing functions of $y$, which violates the definition of a predictive system as a set of predictive distribution functions. The following lemma provides a sufficient condition on the conformity measure $A$ to ensure that the bounds are increasing functions of $y$ \citep{vovk2022book}.
\\
\begin{lemma}
    Suppose that, for any $(x_1, y_1), \dots, (x_n, y_n) \in \Xx \times \R$ and any $x_{n+1} \in \mathcal{X}$, the function
    \begin{equation}\label{eq:mono}
        y \mapsto A(\Pnp(y), (x_{n+1}, y)) - A(\Pnp(y), (x_i, y_i))
    \end{equation}
	is increasing, for all $i = 1, \dots, n$. Then, when $\Pi_{\ell,x_{n+1}}^{\Pn}$ and $\Pi_{u,x_{n+1}}^{\Pn}$ are defined using \eqref{eq:CM_low} and \eqref{eq:CM_upp}, the functions $y \mapsto \Pi_{\ell,x_{n+1}}^{\Pn}(y)$ and $y \mapsto \Pi_{u,x_{n+1}}^{\Pn}(y)$ are increasing.
\end{lemma}

The requirement at \eqref{eq:mono} is referred to as the \emph{monotonicity condition}. If this is satisfied, then the indicator $\one\{A(\Pnp(y),(x_i, y_i)) \leq A(\Pnp(y),(x_{n+1}, y)) \}$ is an increasing function of $y$, and hence the lower and upper bounds of the predictive system at \eqref{eq:CM_low} and \eqref{eq:CM_upp} must also be increasing in $y$. The thickness of the resulting predictive system is then at most $1/(n+1)$.
\\
\begin{remark}\label{rem:CM_mono}
    The monotonicity condition is most relevant in the full conformal setting. In the split conformal framework, the first argument of the conformity measure is fixed at $\PN$. Hence, for conformity measures of the form \eqref{eq:conf}, the parameters of the regression model would be obtained from the estimation data set, in which case \eqref{eq:mono} simplifies to $y \mapsto y - \hat{y}_{n+1} - y_i + \hat{y}_i$ (with $\hat{y}_i = \hat{y}_{x_i}$ for $i = 1, \dots, n+1$). The latter three terms are independent of $y$, so this is trivially an increasing function of $y$. Any conformity measure in this form therefore immediately yields valid predictive distributions and predictive systems in the split conformal setting. The same is not true in the full conformal setting, as elucidated by the following example.
\end{remark}
\bigskip
\begin{example}[Least Squares Prediction Machine]\label{ex:lspm}
    Consider the conformity measure at \eqref{eq:conf} where $\hat{y}_x$ is the least squares prediction for $y$, trained using $(x_1, y_1), \dots, (x_n, y_n)$ and evaluated at $x$. \cite{vovk2017} term the corresponding predictive system the \emph{Least Squares Prediction Machine} (LSPM). However, according to \citet[Proposition 6]{vovk2017}, this conformity measure does not always satisfy the monotonicity condition at \eqref{eq:mono}. Instead, it is common to replace \eqref{eq:conf} with the studentised conformity measure 
    \begin{equation}\label{eq:conf_stu}
        A(\Pn,(x, y)) := \frac{y - \hat{y}_x}{\sqrt{1 - h_{(x,y)}}},
    \end{equation}
    where $h_{(x,y)}$ is the leverage of $(x, y)$ in the regression fit. It is not restrictive to assume that this leverage term exists, since the conformal predictive systems in the full conformal setting only require evaluating the conformity measure at a covariate-outcome pair that is in the support of the first argument of $A$. In the split conformal setting, this studentised conformity measure cannot be used, since the conformity measure must be evaluated at pairs $(x, y)$ that are not in the support. However, from Remark \ref{rem:CM_mono}, the standard non-normalised conformity measure at \eqref{eq:conf} can be employed in this case without issue. The studentised conformity measure at \eqref{eq:conf_stu} satisfies the monotonicity condition, unlike its non-normalised counterpart at \eqref{eq:conf}, and using this within \eqref{eq:CM_low} and \eqref{eq:CM_upp} yields the \emph{studentised LSPM} of \cite{vovk2017} \citep[see also][Section 7.3]{vovk2022book}. Similar results holds for kernel regression \citep{vovk2018a}.
\end{example}

The studentisation of the LSPM is reminiscent of a strategy to \emph{localise} conformal predictive systems \citep{PapadopoulosEtAl2008,LeiEtAl2018}. The conformity measure at \eqref{eq:conf} is often adapted to use a standardised residual, where the residual is divided by an estimate of the uncertainty in the prediction that is similarly obtained from $\Pn$. In doing so, the conformity measure can generally better adapt to regions of the outcome space where the performance of the prediction is less reliable. More generally, we could apply any strictly increasing transformation $\hat{f}_x$ to the residuals, $A(\Pn, (x, y)) = \hat{f}_x(y - \hat{y}_x)$, where $\hat{f}_x$ is obtained from $\Pn$ and $x$. However, in general, there is still no guarantee that these local and transformed conformity measures satisfy the monotonicity condition at \eqref{eq:mono}, and details would therefore have to be worked out for specific cases. In the next section, we introduce predictive systems that coincide with conformal predictive systems in the split conformal setting, and that are more intuitive in the full conformal setting since they naturally avoid the monotonicity condition.

\section{Residual Distribution Predictive Systems}\label{sec:rdps}

\subsection{Definition}

Theorem \ref{thm:intoout} holds for any in-sample calibrated forecasting procedure. An open question is what forecasting procedures are in-sample calibrated. We study the following approach.
\\
\begin{definition}
    Let $\hat{y}_i \in \R$ denote a point-prediction for $y_i$, based on $x_i$ and the training data $\Pn$, for $i = 1, \dots, n$, and denote by $\hat{\varepsilon}_i = y_i - \hat{y}_i$ the corresponding residuals. Let $\hat{y}_x \in \R$ similarly denote a point-prediction for a new outcome, obtained from $\Pn$ and a covariate $x \in \Xx$. The \emph{Residual Distribution} forecasting procedure is defined as
	\begin{equation}\label{eq:RD_G}
	    G_x^{\Pn}(y) = \frac{1}{n}\sum_{i=1}^n \one\{\hat{y}_x + \hat{\varepsilon}_i \leq y\}, \quad y \in \mathbb{R}.
	\end{equation}
    The \emph{Residual Distribution Predictive System} (RDPS) is the predictive system $\Pi$ generated by this forecasting procedure using the bounds defined at \eqref{eq:Gbands}.
\end{definition}

The forecasting procedure at \eqref{eq:RD_G} fits a regression model to the training data, and calculates the in-sample residuals. It then takes an out-of-sample prediction for the new outcome, and dresses this prediction using the training residuals to obtain a discrete predictive distribution. If the regression model yields inaccurate forecasts, then the residuals will generally be large, leading to a highly dispersed distribution, whereas more accurate forecasts will lead to more concentrated predictive distributions. An illustration of this using simulated data is provided in Figure \ref{fig:simdata}.

Just as residuals can be replaced with standardised residuals to obtain a locally adaptive conformity measure, the RDPS can similarly be defined using standardised residuals. More generally, we can again apply any strictly increasing transformation to the residuals.
\\
\begin{definition}
    Let $\hat{y}_i \in \R$ denote a prediction for $y_i$, and let $\hat{f}_i : \R \to \R$ be a strictly increasing function, based on $x_i$ and the training data $\Pn$, for $i = 1, \dots, n$. Denote by $\hat{\varepsilon}_i = y_i - \hat{y}_i$ the corresponding residuals. Let $\hat{y}_x \in \R$ similarly denote a point-prediction for a new outcome, and let $\hat{f}_x : \R \to \R$ be a strictly increasing function, both obtained from $\Pn$ and a covariate $x \in \Xx$. The \emph{generalised Residual Distribution} forecasting procedure is defined as
	\begin{equation}\label{eq:gen_RDPS}
    G_x^{\Pn}(y) = \frac{1}{n}\sum_{i=1}^n \one \left\{ \hat{y}_x + \hat{f}_x^{-1}(\hat{f}_i(\hat{\varepsilon}_i)) \leq y \right\}, \quad y \in \mathbb{R}.  
	\end{equation}
    The predictive system generated by \eqref{eq:gen_RDPS} is referred to as the \emph{generalised RDPS}.
\end{definition}

This generalised RDPS similarly calculates the in-sample residuals of a regression model, but transforms them according to the functions $\hat{f}_i$ before dressing the new prediction $\hat{y}_x$. The motivation behind this approach is that different residuals in the training data set can be rescaled depending on their relevance to the new prediction. The functions $\hat{f}_i$ therefore help to incorporate context-dependent information into the predictive distribution and, in turn, the predictive system.
\\
\begin{example}\label{ex:std_RDPS}
    Let $\hat{f}_i(t) = t/\hat{\sigma}_i$, where $\hat{\sigma}_i > 0$ is a measure of uncertainty in the prediction $\hat{y}_i$, for $t \in \R$ and $i = 1, \dots, n$. Let $\hat{y}_x$ be a prediction for a new outcome, based on the covariate $x$, and let $\hat{f}_x(t) = t/\hat{\sigma}_x$, where $\hat{\sigma}_x >0$ is similarly a measure of uncertainty in $\hat{y}_x$. The generalised RDPS at \eqref{eq:gen_RDPS} becomes 
	\[
    G_x^{\Pn}(y) = \frac{1}{n}\sum_{i=1}^n \one \left\{ \hat{y}_x + \frac{\hat{\sigma}_x}{\hat{\sigma}_i} \hat{\varepsilon}_i \leq y \right\}, \quad y \in \mathbb{R}.
    \]
    This corresponds to an approach whereby the residuals are standardised by a measure of uncertainty in the prediction, before being added to the new prediction $\hat{y}_x$.
\end{example}

\citet[Example 3]{allen2025} demonstrate that the forecasting procedure at \eqref{eq:RD_G} is in-sample probabilistically calibrated, and this extends easily to the generalised forecasting procedure at \eqref{eq:gen_RDPS}. The resulting RDPS is therefore guaranteed to contain an out-of-sample calibrated forecast distribution for $Y_{n+1}$ (under the assumption of exchangeability). Hence, this approach shares the same calibration properties as the classic conformal predictive systems. The RDPS has the advantage that it generates valid predictive distributions and predictive systems for any choice of the regression method; it does not require that the fairly stringent monotonicity condition at \eqref{eq:mono} is satisfied, for example. This facilitates the introduction of predictive systems with out-of-sample calibration guarantees using any regression procedure. However, for some choices of the regression method in the full conformal setting, the thickness of the RDPS will be excessively high, in which case the predictive systems are not particularly useful in practice. In contrast, conformal predictive systems have a thickness of at most $1/(n+1)$.

\subsection{Comparison with Conformal Predictive Systems}\label{sec:comp}

Consider first the split conformal setting. In this case, the following theorem demonstrates that Residual Distribution Predictive Systems are actually equivalent to conformal predictive systems with conformity measure of the form \eqref{eq:conf} when both are constructed using the same regression model. This equivalence holds more generally for the generalised RDPS and the conformal predictive systems defined using the transformed residuals. 

\begin{figure}
    \centering
    \includegraphics[width=0.32\linewidth]{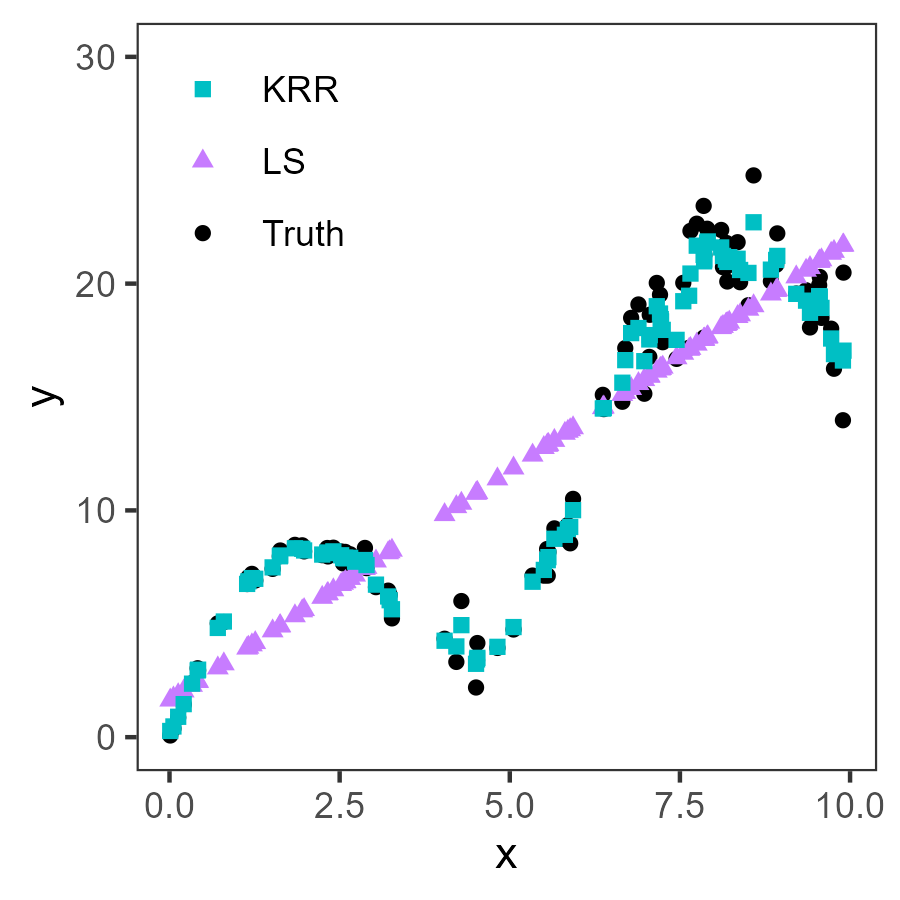}
    \includegraphics[width=0.32\linewidth]{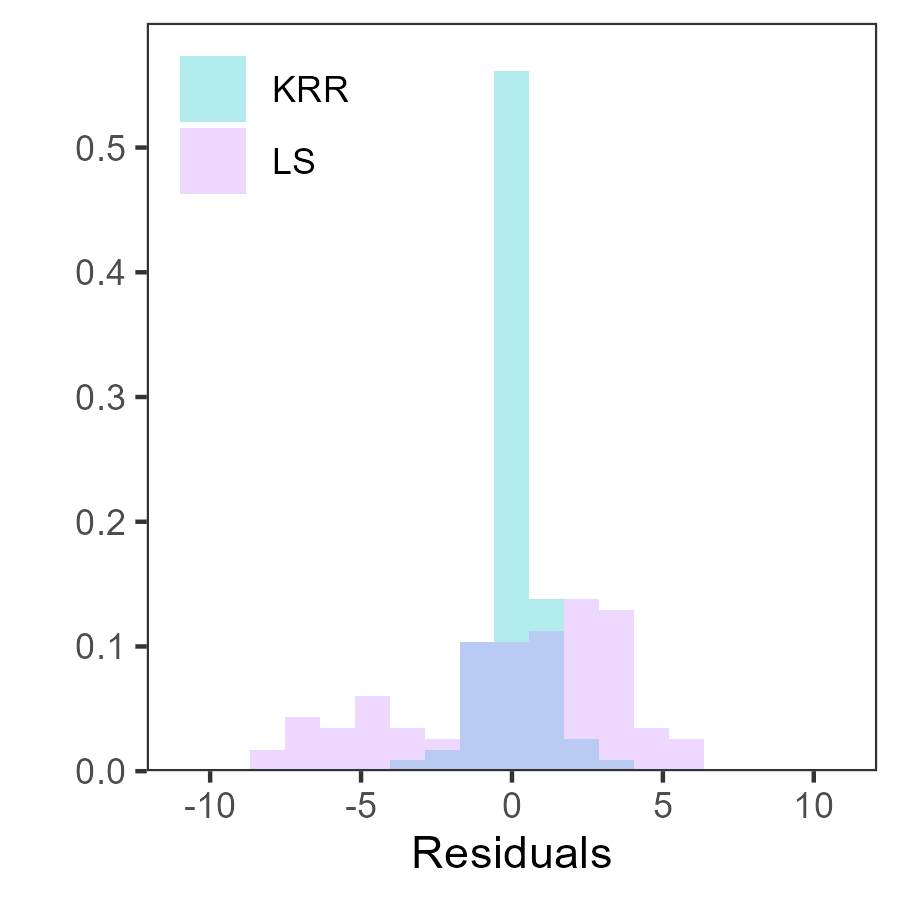}
    \includegraphics[width=0.32\linewidth]{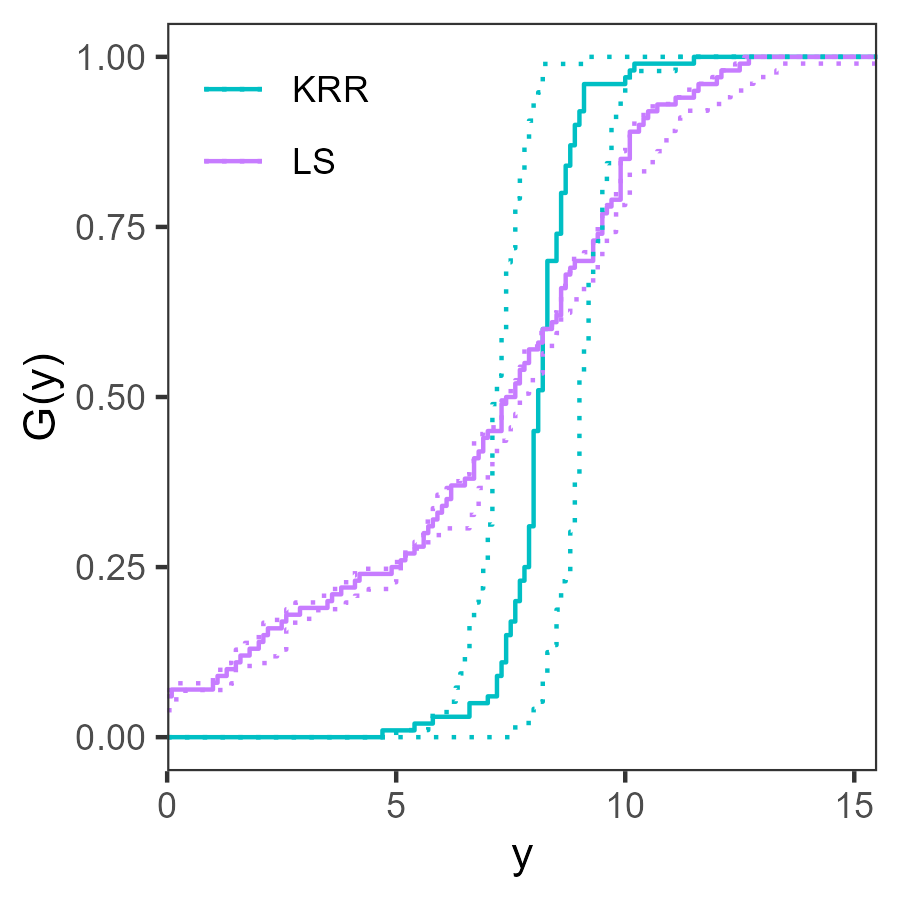}
    \caption{Left: 100 samples from the nonlinear simulation data in Section \ref{sec:simstudy}, along with the point forecasts obtained from ordinary least squares regression (LS) and kernel ridge regression (KRR). Centre: A histogram of the corresponding residuals. Right: The resulting predictive distributions $G_{x_{n+1}}^{\Pn}$ (solid) and RDPS bounds $\Pi_{\ell, x_{n+1}}^{\Pn}, \Pi_{u, x_{n+1}}^{\Pn}$ (dotted) at $x_{n+1} = 2.5$. Least squares prediction yields less accurate forecasts for the nonlinear data than kernel ridge regression, resulting in higher residuals and a more dispersed predictive distribution. However, the RDPS thickness is larger for kernel ridge regression since the model is more flexible and therefore more difficult to estimate.}
    \label{fig:simdata}
\end{figure}

\begin{theorem}\label{thm:equiv}
    Suppose that, for any $x \in \Xx$, we have a prediction $\hat{y}_x \in \R$ and a strictly increasing function $\hat{f}_x : \R \to \R$ based on $\PN$. Define a conformity measure by $A(\PN, (x, y)) = \hat{f}_x(y - \hat{y}_x)$. Then, in the split conformal setting, the conformal predictive system defined using this conformity measure is equivalent to the generalised RDPS defined by \eqref{eq:gen_RDPS}.
\end{theorem}

\begin{proof}
    Denote $\hat{f}_i = \hat{f}_{x_i}$, $\hat{y}_i = \hat{y}_{x_i}$, and $\hat{\varepsilon}_i = y_i - \hat{y}_i$, for $i = N+1, \dots, n + 1$. In the split conformal setting, the forecasting procedure at \eqref{eq:gen_RDPS} becomes   
    \begin{align*}
    G_{x_{n+1}}^{\Pn}(y) &= \frac{1}{n - N} \sum_{i = N+1}^{n} \one\{\hat{y}_{n+1} + \hat{f}_{n+1}^{-1}(\hat{f}_i(\hat{\varepsilon}_i)) \leq y\}, \quad y \in \R \\
    &= \frac{1}{n - N} \sum_{i = N+1}^{n} \one\{\hat{f}_i(y_i - \hat{y}_i)) \leq \hat{f}_{n+1}(y - \hat{y}_{n+1})\}, \quad y \in \R.
    \end{align*}
    This is exactly the forecasting procedure at \eqref{eq:CM_G} with $A(\PN, (x, y)) = \hat{f}_x(y - \hat{y}_x)$. Since the residual distribution and conformity measure predictive systems both correspond to the bounds at \eqref{eq:Gbands} defined using the same forecasting procedure, the predictive systems themselves must also be equivalent.
\end{proof}

Theorem \ref{thm:equiv} holds for any regression method used to obtain the predictions, and any choice of the strictly increasing functions. The forecasting procedure underlying the RDPS at \eqref{eq:CM_G} is obtained when $\hat{f}_x$ is the identity function for all $x \in \Xx$, while the standardised RDPS in Example \ref{ex:std_RDPS} is recovered using the functions $\hat{f}_x(t) = t/\hat{\sigma}_x$ for $t \in \R$. 

One immediate corollary of Theorem \ref{thm:equiv} is that the thickness of the RDPS in the split conformal framework is equal to $1/(n - N + 1)$, where the denominator is simply the size of the calibration dataset plus one. This thickness is constant and does not depend on $y$, in contrast to the full conformal framework, where the thickness of the RDPS varies with $y$ (see Section \ref{sec:simstudy}). Hence, while the conformal and Residual Distribution predictive systems are equivalent in the split conformal setting, they differ when a full conformal framework is adopted. One benefit of the RDPS is that they do not require that any monotonicity condition is satisfied, as in \eqref{eq:mono}. In principle, this allows the approach to be implemented with any regression method. 

However, for some regression methods, the bounds at \eqref{eq:Gbands} will be too wide to be informative. We can make the following approximation to find a sufficient criterion for informative bounds. Let $\hat{y}_i'$ denote the prediction for $y_i$ when the regression method is trained using $\Pnp(y')$, for $y' \in \R$. The superscript prime in $\hat{y}_i'$ is used to clarify that this is a function of $y'$. Then, 
\begin{align*}
    \Pi^{\Pn}_{u,x_{n+1}}(y) & = \sup_{y' \in \R} G_{x_{n+1}}^{\Pnp(y')}(y) \\
    &= \sup_{y' \in \R} \frac{1}{n+1}\left(\sum_{i=1}^{n}\one\{\hat{y}_{n+1}' + y_i - \hat{y}_i' \le y\} + \one\{\hat{y}_{n+1}' + y' - \hat{y}_{n+1}' \le y\}\right)\\
     &= \sup_{y' \in \R} \frac{1}{n+1}\left(\sum_{i=1}^{n}\one\{\hat{y}_{n+1}' - \hat{y}_i' \le y - y_i\} + \one\{y' \le y \}\right)\\
    &\le  \frac{1}{n+1}\left(\sum_{i=1}^{n}\one\{\inf_{y' \in \R}(\hat{y}_{n+1}' - \hat{y}_i') \le y - y_i\} + 1\right).
\end{align*}
Similarly, 
\[
    \Pi^{\Pn}_{\ell,x_{n+1}}(y) = \inf_{y' \in \R} G_{x_{n+1}}^{\Pnp(y')}(y) \ge  \frac{1}{n+1}\sum_{i=1}^{n}\one\{\sup_{y' \in \R}(\hat{y}_{n+1}' - \hat{y}_i') \le y - y_i\}.
\]

Hence, for the RDPS to have reasonable thickness, we want that the regression method is such that $\inf_{y' \in \R}(\hat{y}_{n+1}' - \hat{y}_i')$ and $\sup_{y' \in \R}(\hat{y}_{n+1}' - \hat{y}_i')$ are bounded. For example, when using quantile regression to determine $\hat{y}_i'$, we have the property that the regression line, or the regression hyperplane more generally for $p$-dimensional covariates, passes through at least $p$ data points if it is unique. When $y'$ is extreme, it is highly unlikely that the regression hyperplane would pass through $(x_{n+1},y')$, so the difference $\hat{y}_{n+1}' - \hat{y}_i'$ is robust to outlying values of $y'$. This ensures informative predictive bounds. In contrast, for ordinary least squares regression or kernel regression, the difference $\hat{y}_{n+1}' - \hat{y}_i'$ is typically linear in $y'$, so the bounds will typically be uninformative unless the range of $y_i$ is bounded.
\\
\begin{remark}[Deleted RDPS]\label{rmk:delete}
    One approach to obtain robust regression methods for this purpose is to simply identify and remove outliers from the augmented training data. This shares similarities with the \emph{deleted} LSPM proposed by \cite{vovk2017}, though here we cannot simply remove the pair $(x_{n+1}, y)$ since Theorem \ref{thm:intoout} requires that the parameter estimation procedure is permutation invariant. Instead, one could fit the regression model to the training data, and then remove the covariate-outcome pairs that yield large (absolute) residuals. The regression model can then be refit to the remaining training data. This increases the computational cost, but circumvents the sensitivity of the RDPS to outlying values. This approach is implemented in Section \ref{sec:simstudy} when calculating the residuals using least squares regression and kernel ridge regression. Several other approaches exist to construct robust regression models \citep[e.g.][]{SalibianBarrera2023}, and they could similarly be employed.
\end{remark}

\begin{figure}[t]
	\centering
    \includegraphics[width=0.5\textwidth]{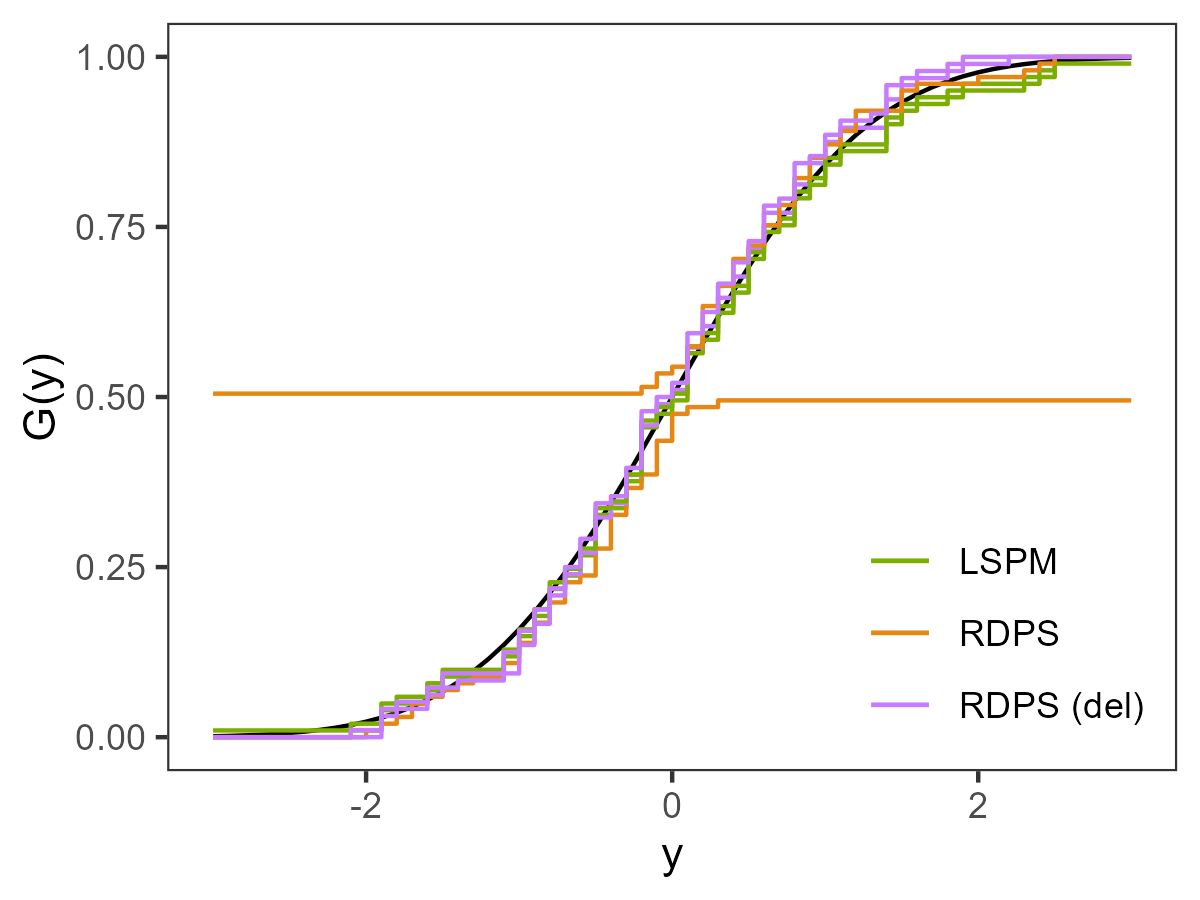}
	\caption{Lower and upper bounds of the LSPM, RDPS, and the deleted RDPS (see Remark \ref{rmk:delete}) at a given $x_{n+1}$ for the linear simulation data in Section \ref{sec:simstudy} Both RDPS models are obtained using ordinary least squares regression. The black line is the true conditional distribution of the outcome at this covariate value.}
	\label{fig:ss_full}
\end{figure}

\subsection{Computation}

Consider now the computation of the RDPS. The bounds at \eqref{eq:Gbands} correspond to infimum and supremum of the predictive distributions obtained for all possible $y' \in \R$. In practice, these could be calculated by discretising the real line, refitting the forecasting procedure for all $y'$ on this discretised line, and then taking the pointwise minimum and maximum of the resulting predictive distributions. However, this approach is generally prohibitively computationally expensive, and the results may also depend on how the real line is discretised, particularly in the limits. Moreover, in practice, it is often not necessary to consider every possible value of $y'$. 

For the standard RDPS, we have that
\begin{equation}\label{eq:RDPS_G_full}
    G_{x_{n+1}}^{\Pnp(y')}(y) = \frac{1}{n + 1} \left[ \sum_{i=1}^n \one\{\hat{y}'_{n+1} - \hat{y}'_i \leq y - y_i\} + \one\{y' \le y\} \right]. 
\end{equation}
In the split conformal setting, the predictions $\hat{y}'_{n+1}$ and $\hat{y}'_i$ do not depend on $y'$, so the lower and upper bounds can be calculated by setting $\one\{y' \le y\}$ to zero and one, respectively. Similarly, in the full conformal setting, if $\hat{y}'_{n+1} - \hat{y}'_i$ is an increasing function of $y'$, for all $i = 1, \dots, n$, then \eqref{eq:RDPS_G_full} is a decreasing function of $y'$. In this case, the lower and upper bounds of the predictive system can be computed efficiently using just two fits of the forecasting procedure: the infimum is obtained as $y' \to +\infty$, and the supremum as $y' \to -\infty$. Some examples where this is (and is not) satisfied are given in Appendix \ref{app:comp}.

Even if $\hat{y}'_{n+1} - \hat{y}'_i$ is not increasing in $y'$, it may still be possible to efficiently implement the RDPS. This is possible, for example, when the predictions are linear functions of the observations in the training data. Some examples where this is the case include least squares regression, kernel regression, smoothing splines, $k$-nearest neighbours, and binning procedures. \cite{vovk2017} and \cite{vovk2018a} use this to derive efficient algorithms to calculate the LSPM and its kernel regression counterpart, the Kernel Ridge Regression Prediction Machine (KRRPM). Details to calculate the RDPS in this case are described in Appendix \ref{app:comp}. Alternatively, other estimation methods could be leveraged to obtain more efficient implementations of the RDPS. We provide an example below for quantile regression; details for other methods would have to be similarly worked out.
\\
\begin{example}[Quantile regression] 

    A robust alternative to ordinary least squares regression is median regression (more generally, quantile regression) \citep{KoenkerBassett1978}. \cite{ShenEtAl2024} recently proposed an iterative algorithm for parameter estimation in online quantile regression. We modify this slightly to obtain an approach that allows for efficient estimation of the RDPS bounds when the predictions are obtained using quantile regression. Firstly, we randomise the order of the training pairs $(x_1, y_1), \dots, (x_n, y_n), (x_{n+1}, y')$, and let $(x_{(i)}, y_{(i)})$ denote the $i$-th pair in the randomised sequence, for $i = 1, \dots, n+1$. Following \cite{ShenEtAl2024}, an initial parameter estimate $\beta_{(0)}$ is then iteratively updated using sub-gradient descent, with
    \[
    \beta_{(t+1)} = \beta_{(t)} - \eta_{t+1} \left( \tau - \one \big\{ y_{(t+1)} > \beta_{(t)}^{\top} x_{(t+1)} \big\} \right) x_{(t+1)}, \quad t = 0, \dots, n,
    \]
    where $\eta_1, \dots, \eta_{n+1} > 0$ are stepsize parameters, and $\tau \in (0, 1)$ is the quantile level of interest. We assume that the parameters $\beta_{(0)}$ and $\eta_1, \dots, \eta_{n+1}$ do not depend on $y'$. We then take the final predictions as $\hat{y}_i = \beta_{(n+1)}^{\top}x_{i}$, for $i = 1, \dots, n + 1$. All predictions are therefore derived from the final parameter in the algorithm, $\beta_{(n+1)}$. Assume that $(x_{n+1}, y') = (x_{(j)}, y_{(j)})$, i.e. $(x_{n+1}, y')$ is the $j$-th training pair in the randomised sequence for $j \in \{1, \dots, n+1\}$. Then, the parameter $\beta_{(n+1)}$, and hence the predictions $\hat{y}_i'$, only depend on $y'$ via the indicator $\one \big\{ y' > \beta_{(j-1)}^{\top} x_{n+1} \big\}$. Hence, over all possible $y' \in \R$, $\hat{y}_{n+1}' - \hat{y}_i'$ takes on only two possible values, meaning the infimum and supremum of the forecasting procedure at \eqref{eq:RDPS_G_full} can be obtained from just two runs of the iterative parameter estimation, which is anyway quick to implement. This can be extended further to run the iterative algorithm multiple times, and to average the resulting parameter estimates; while the RDPS is guaranteed to contain a calibrated probabilistic prediction for $Y_{n+1}$ for any regression model, better parameter estimates will result in more informative predictive systems.

\end{example}

\section{Simulation Examples}\label{sec:simstudy}

We compare the performance of the RDPS and conformal predictive systems when applied to two simple simulated datasets. The first case assumes that there is a linear relationship between the covariates and the outcome variable: $X_i \sim \mathcal{N}(0, 1)$, and $Y_i \sim \mathcal{N}(X_i, 1)$ for $ i = 1, \dots, n + 1$. The second case assumes that there is a nonlinear relationship between $X_i$ and $Y_i$, and that the conditional variance of the outcome is a function of the covariate: $X_i \sim \text{Uniform}(0, 10)$ and $Y_i = 2X_i + 5\sin(X_i) + \eta_i$, where $\eta_i \sim \mathcal{N}(0, (x_i/5)^2)$. In both cases, $(X_1,Y_1),\dots,(X_{n+1},Y_{n+1})$ are iid. A sample from the nonlinear simulated data is shown in Figure \ref{fig:simdata}.

In each case, we set $n = 100$, and draw realisations of $(X_1,Y_1),\dots,(X_{n+1},Y_{n+1})$. We fit a conformal predictive system and an RDPS to the training data and $X_{n+1}$, and assess the predictive system in its ability to predict $Y_{n+1}$. This is repeated 1000 times. The predictive systems are compared via the (conservative) prediction intervals derived from them. We evaluate the $1-\alpha$-level prediction intervals for $1-\alpha$ ranging from 0.5 to 0.95. The central $1-\alpha$-level prediction interval is obtained from the predictive system using the $\alpha/2$-quantile of $\Pi_u$ and the $1 - \alpha/2$-quantile of $\Pi_\ell$. While these bounds may be defective distribution functions, these quantiles exist in all cases here. Since the predictive systems contain a probabilistically calibrated predictive distribution out-of-sample, the resulting prediction intervals contain $Y_{n+1}$ with probability $\ge 1 - \alpha$. The prediction intervals are then assessed with respect to their unconditional coverage, average width, and average interval score; the interval score is a popular proper scoring rule to compare prediction intervals, and a lower score is desired \citep{GneitingRaftery2007}. 

Since conformal predictive systems and the RDPS are equivalent in the split conformal setting, we restrict attention to the full conformal setting. Conformal predictive systems with a conformity measure of the form \eqref{eq:conf_stu} have been proposed using ordinary least squares regression \citep[the LSPM;][]{vovk2017} and kernel ridge regression \citep[the KRRPM;][]{vovk2018a}. For comparison, we therefore implement these two regression models within the RDPS. Following \cite{vovk2018a}, we employ kernel ridge regression using the Laplacian (exponential) kernel, with the regularisation parameter estimated using cross-validation on an independent sample of the simulated data. We use the deleted versions of the RDPS, as discussed in Remark \ref{rmk:delete}.

While the conformal predictive systems have a fixed thickness of $1/(n + 1)$, the thickness of the RDPS changes as a function of the covariate. Large covariates typically lead to larger thicknesses, aligning with the idea that the thickness provides a measure of epistemic forecast uncertainty. The distribution of the thickness of the RDPS over the 1000 repetitions is shown in Figure \ref{fig:ss_th}. Since kernel ridge regression is more flexible than ordinary least squares, it typically results in a larger thickness. This is especially true in the nonlinear case.

\begin{figure}[t]
	\centering
    \includegraphics[width=0.33\textwidth]{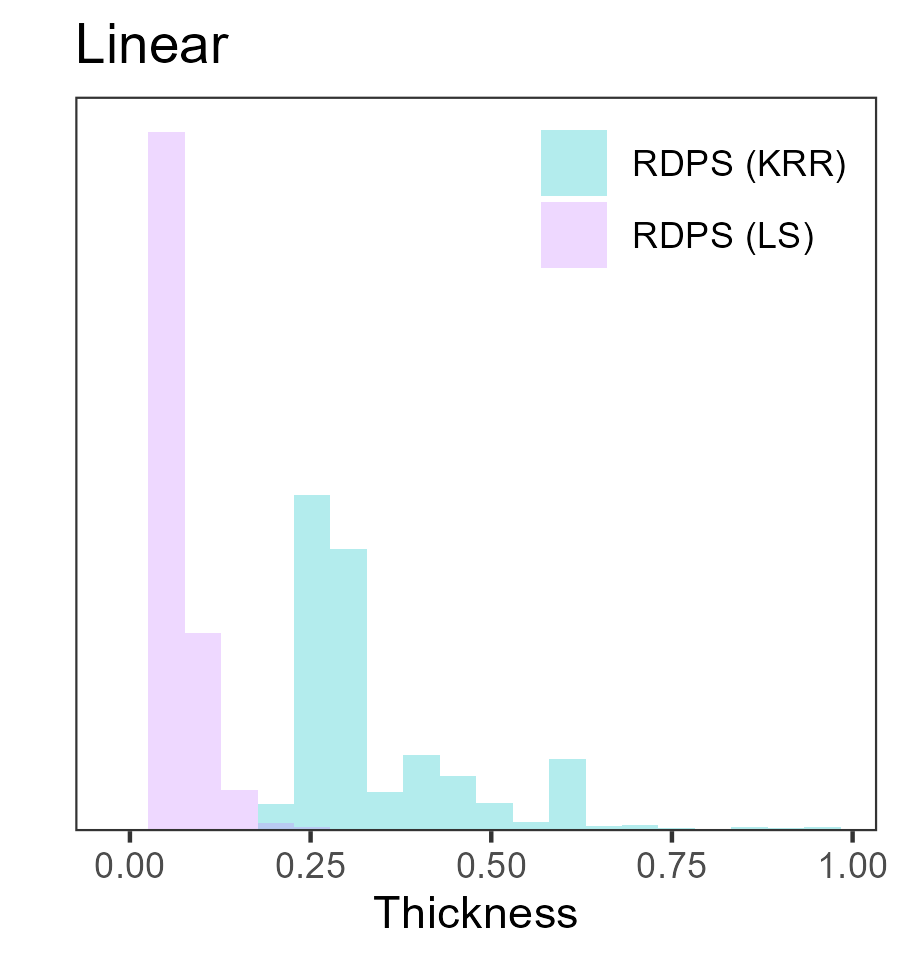}
    \includegraphics[width=0.33\textwidth]{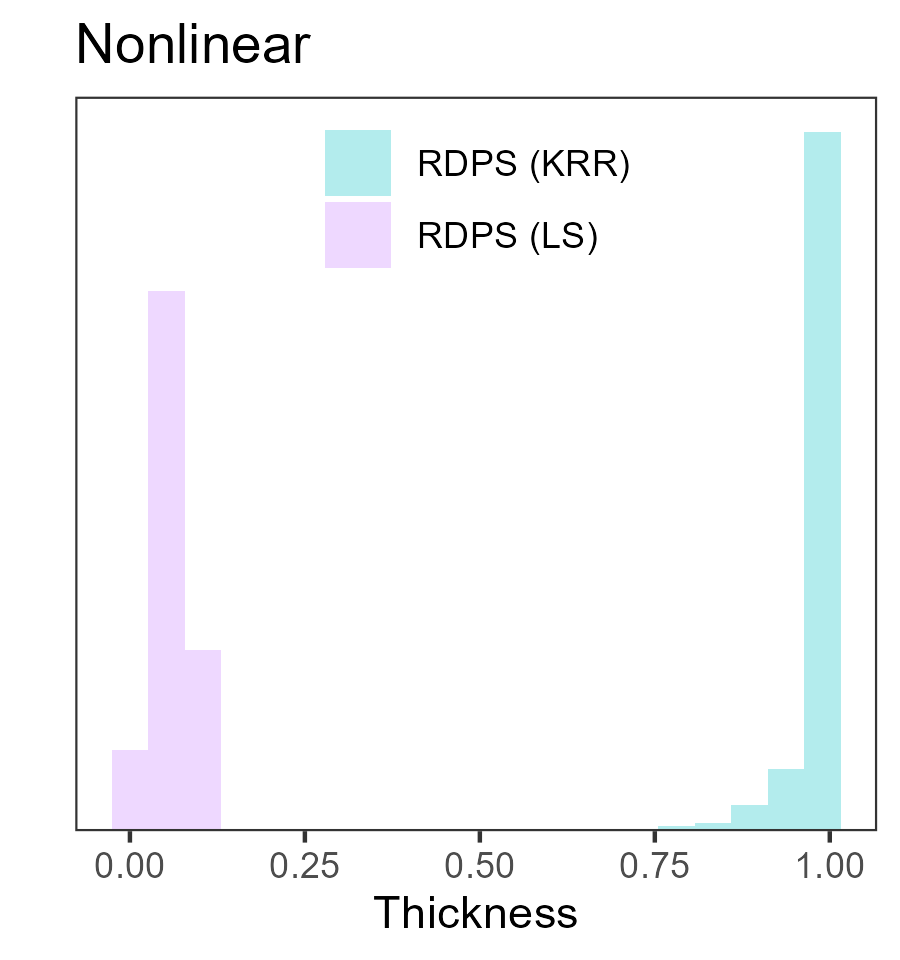}
	\caption{Histograms of the thickness of the four predictive systems.}
	\label{fig:ss_th}
\end{figure}

The coverage, width, and interval score of the four predictive systems are shown in Figure \ref{fig:ss_full_eval_lin} for the linear case. The unconditional coverage of the prediction intervals is approximately equal to the desired coverage level in all cases. Due to the larger thickness of the RDPS with kernel ridge regression, this approach yields slightly wider, more conservative prediction intervals when the level $1 - \alpha$ is close to 0.5. Nonetheless, the width and coverage of all intervals is very similar for the four methods, and they therefore all result in similar interval scores. The more parsimonious least squares approaches slightly outperform their kernel ridge regression counterparts, but there is very little to distinguish the RDPS from the classic conformal predictive systems.

\begin{figure}[b!]
	\centering
	\includegraphics[width=0.32\textwidth]{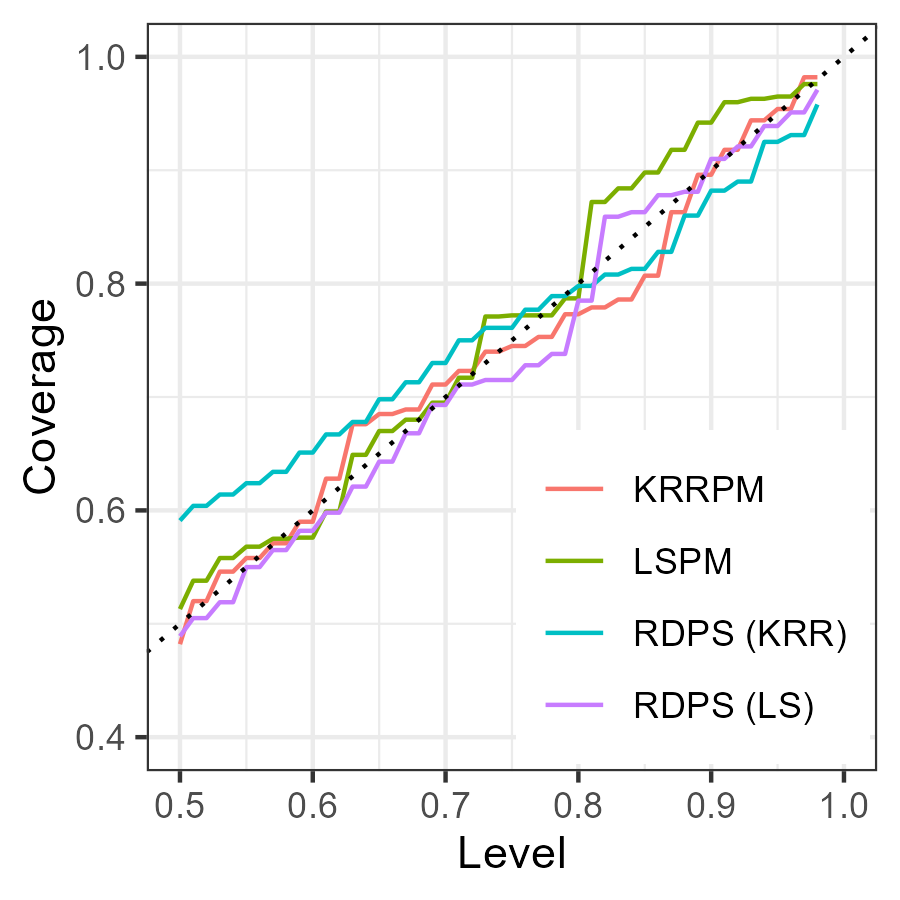}
    \includegraphics[width=0.32\textwidth]{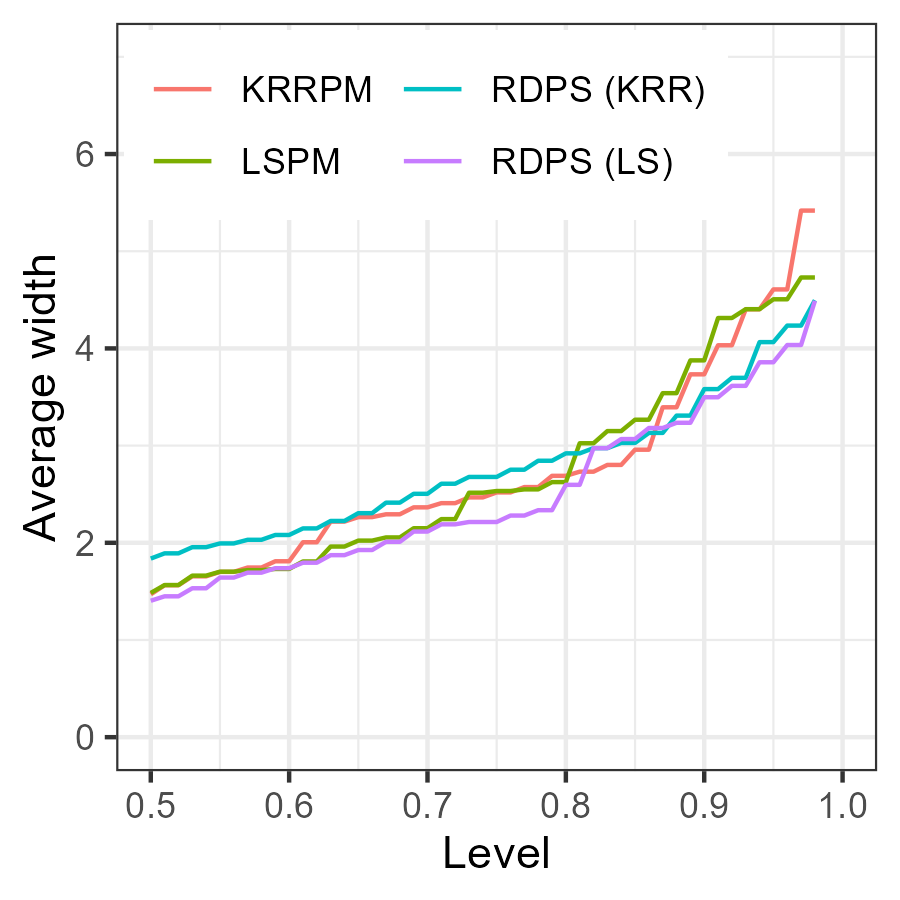}
    \includegraphics[width=0.32\textwidth]{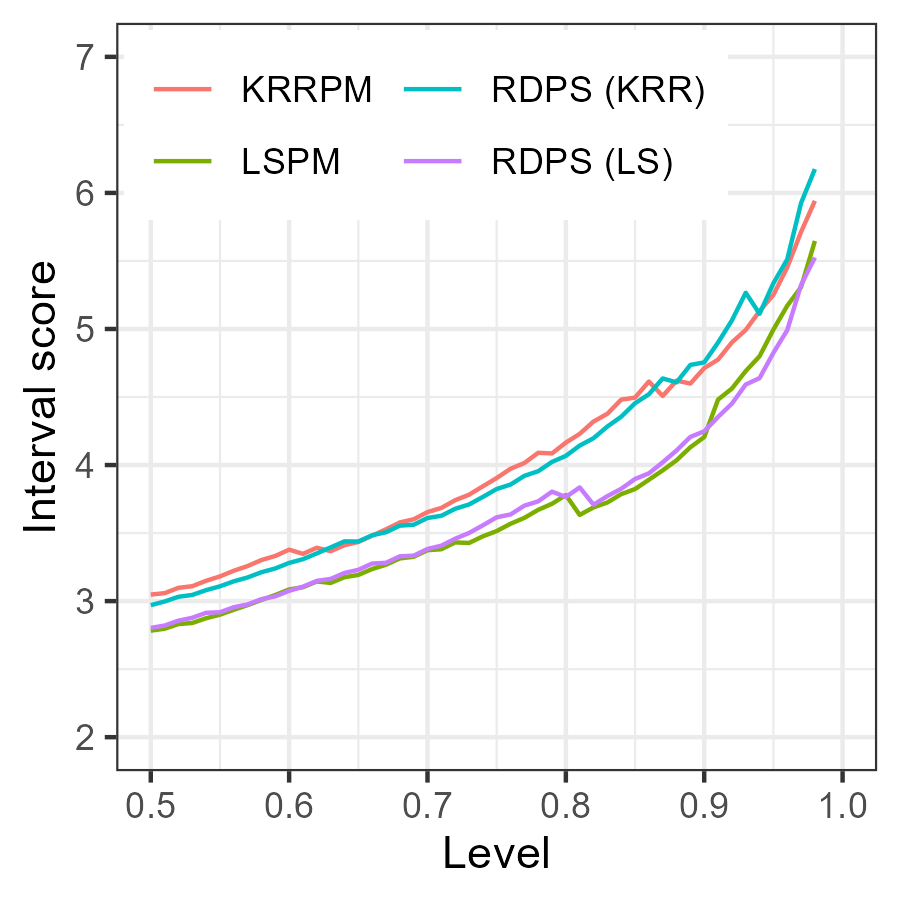}
	\caption{Coverage, average width, and average interval score of prediction intervals obtained from the three predictive systems, shown as a function of the level of the prediction interval. Results are shown for the linear simulated data.}
	\label{fig:ss_full_eval_lin}
\end{figure}

The analogous results are shown for the nonlinear case in Figure \ref{fig:ss_full_eval_non}. The conservativeness of the kernel ridge regression RDPS intervals is even more prevalent in the nonlinear case, owing again to its increased thickness. The kernel ridge regression intervals are nonetheless generally shorter than those obtained using least squares (with either type of predictive system). This is because the least squares predictions do not capture the nonlinearity in the data (Figure \ref{fig:simdata}), leading to highly-dispersed predictive distributions, and thus wide prediction intervals. As a result, there is a clear difference between the interval scores assigned to the prediction intervals obtained from least squares regression and kernel ridge regression. Nonetheless, for both choice of regression model, the results are again similar for conformal predictive systems and the RDPS, with conformal predictive systems offering a slight advantage at most prediction levels. 

While we restrict attention to least squares and kernel ridge regression, since these facilitate a comparison with the LSPM and KRRPM in the classic conformal prediction framework, the RDPS could also be implemented with more flexible regression models. This offers a distinct advantage over conformal predictive systems that only allow for regression models satisfying the monotonicity condition at \eqref{eq:mono}. 

\begin{figure}
	\centering
    \includegraphics[width=0.32\textwidth]{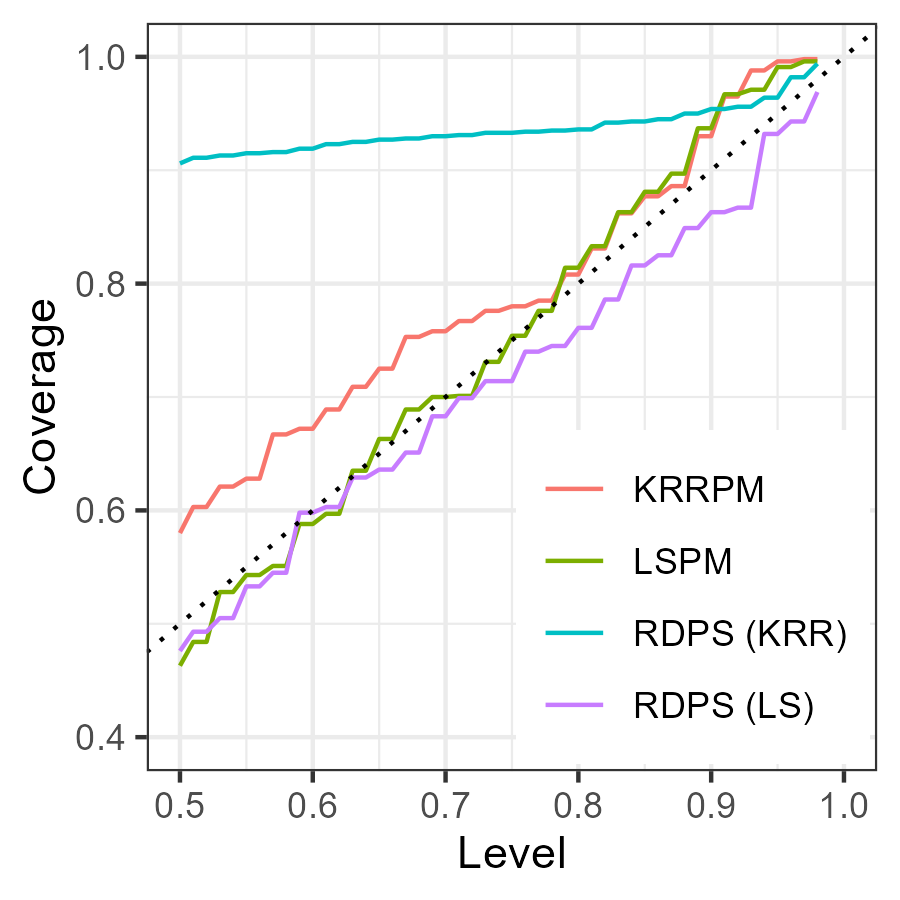}
    \includegraphics[width=0.32\textwidth]{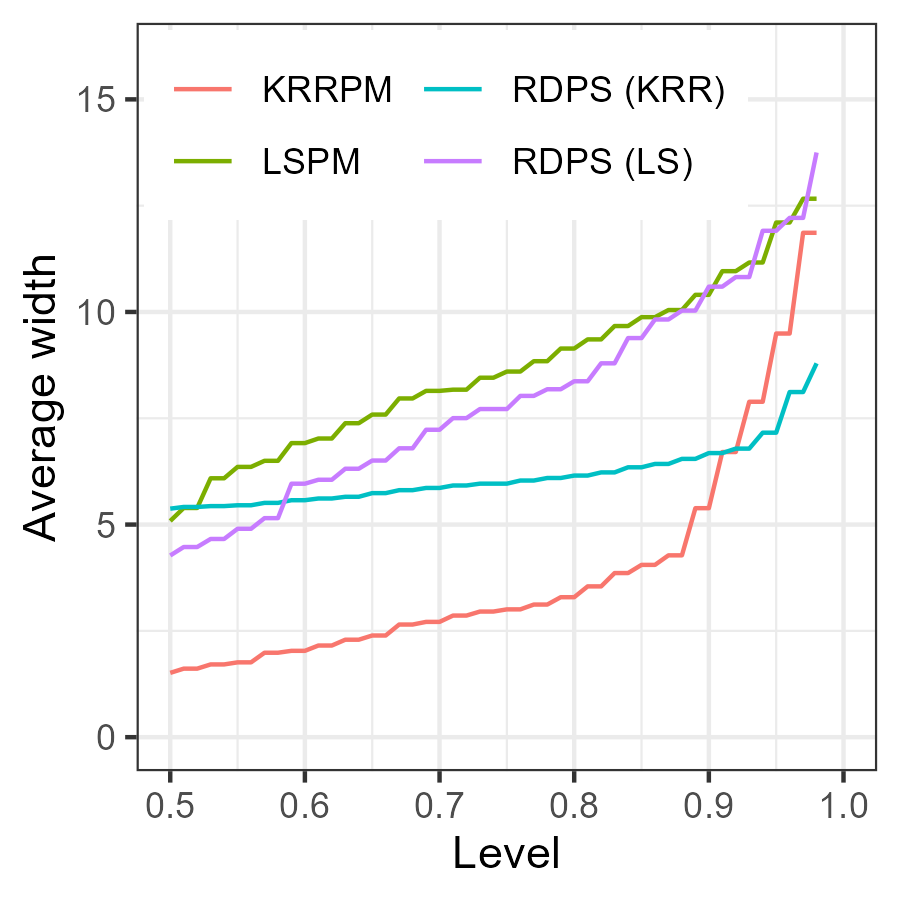}
    \includegraphics[width=0.32\textwidth]{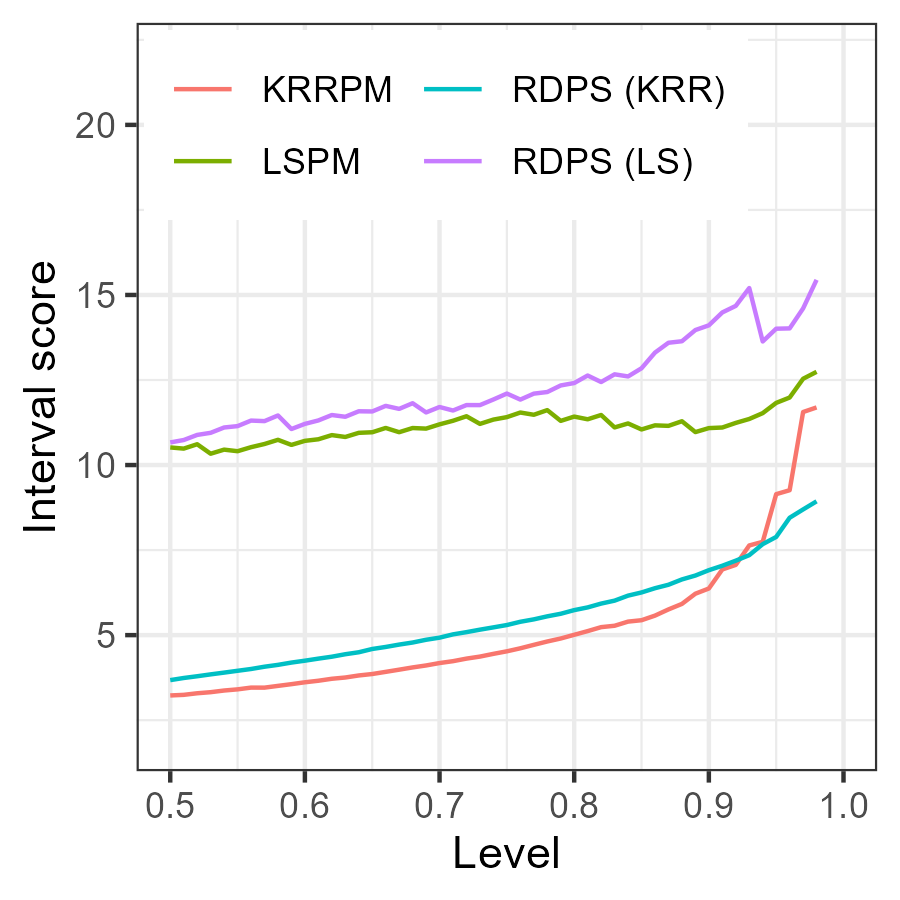}
	\caption{As in Figure \ref{fig:ss_full_eval_lin} but for the nonlinear simulated data.}
	\label{fig:ss_full_eval_non}
\end{figure}

\section{Conclusions}\label{sec:conc}

Conformal prediction provides a means to obtain prediction sets with out-of-sample calibration guarantees. It has recently been shown that conformal predictive systems can be interpreted as one example of a more general framework to construct out-of-sample calibrated predictive systems (i.e. sets of probabilistic forecasts for real-valued outcomes). Alternative approaches have therefore been proposed that similarly construct predictive systems that are guaranteed to contain a calibrated forecast distribution out-of-sample, possibly satisfying a stronger notion of calibration. In this work, we rigorously analyse Residual Distribution Predictive Systems, which are predictive systems obtained via a simple forecasting procedure based on a point-valued regression method and the empirical distribution of its residuals in the training data.

We demonstrate that this approach coincides with conformal predictive systems in the split conformal setting, facilitating a better understanding of classical conformal prediction. This equivalence holds for the typical conformity measures used in practice, based on (transformed) residuals of regression methods. The two approaches differ in the full conformal setting, where the validity of conformal predictive systems depends on the conformity measure satisfying a relatively stringent monotonicity condition. In contrast, the new approach can directly be applied with any method, facilitating the design of predictive systems based on more powerful forecasting methods, such as random forests and neural networks. However, while this approach can be applied alongside any regression method, we find that the thickness of the resulting predictive system can be prohibitively large in practice. This can be circumvented by using robust regression methods.

Moreover, the thickness of the residual distribution predictive systems will generally change depending on the model choice and covariate information. In contrast, the thickness of conformal predictive systems is fixed at $1/(n + 1)$. It is often suggested that the size of the predictive system provides a measure of epistemic forecast uncertainty, and since the thickness of conformal predictive systems does not depend on the covariate information, it suggests that this measure of epistemic uncertainty only accounts for the epistemic uncertainty induced by the sample size. Further work is still needed to assess estimates of this second order forecast uncertainty.

Finally, conformal as well as Residual Distribution Predictive Systems satisfy a fairly weak notion of forecast calibration. While this notion is still useful in practice, other more complex methods have been proposed that generate predictive systems with stronger out-of-sample calibration guarantees. However, these methods typically require large amounts of data, whereas the two approaches discussed herein are more parsimonious, rendering them easy to implement even for small sample sizes. It is in these low data regimes that the RDPS is most beneficial; when data is rich, a split conformal approach can be applied, in which case the novel approach is generally equivalent to conformal predictive systems. Overall, conformal predictive systems remain a valid choice in standard settings, especially when monotonic conformity measures are applicable. Residual Distribution Predictive Systems, on the other hand, offer a more flexible framework that accommodates a broader range of regression techniques, making it particularly suited to scenarios involving complex or nonstandard data structures. Future research could therefore further explore the integration of other regression techniques into this framework.

\bibliographystyle{apalike}
\bibliography{references.bib} 

\appendix

\section{Computation}\label{app:comp}

In this appendix, we provide examples of regression models where the prediction difference $\hat{y}'_{n+1} - \hat{y}_i'$ is an increasing function of $y'$, for all $i = 1, \dots, n$, possibly under some conditions. 
\\
\begin{example}[OLS regression]
    Consider an ordinary least squares regression model that issues predictions $\hat{y}_i'$ given training data $(x_1, y_1), \dots, (x_n, y_n), (x_{n+1}, y')$. Let $H$ denote the hat matrix of the regression fit, with elements $h_{i,j}$, for $i,j = 1, \dots, n+1$. Up to an additive constant that does not depend on $y'$, we have
    \[
        \hat{y}'_{n+1} - \hat{y}'_i = (h_{n+1,n+1} - h_{i,n+1})y'.
    \]
    This difference is therefore an increasing function of $y'$ whenever $h_{n+1,n+1} \ge h_{i,n+1}$ for all $i = 1, \dots, n$. Since $h_{i,n+1} \in [-0.5, 0.5]$, this is satisfied if $h_{n+1,n+1} \ge 0.5$. In contrast, \citet[Proposition 7.4]{vovk2022book} states that the LSPM satisfies the monotonicity condition if $h_{n+1,n+1} < 0.5$. An analogous result holds for kernel regression.
\end{example}
\bigskip
\begin{example}[Kernel-density estimation]
    Consider a kernel-density prediction method that, given training data $(x_1, y_1), \dots, (x_n, y_n), (x_{n+1}, y')$, issues predictions
    \[
    \hat{y}'_i = \frac{\sum_{j=1}^{n} \exp(-\|x_i - x_j\|) y_j + \exp(-\|x_i - x_{n+1}\|)y'}{\sum_{j=1}^{n+1} \exp(-\|x_i - x_j\|)}.
    \]
    When predicting $y_i$, this essentially issues a weighted average of the observations in the training data, where the weights decrease exponentially as the covariate moves away from $x_i$. A fixed bandwidth parameter could also be used in the exponentials without changing the following analysis. Up to an additive constant that does not depend on $y'$, we have
    \begin{align*}
        \hat{y}'_{n+1} - \hat{y}'_i &= y' \left\{ \frac{1}{\sum_{j=1}^{n+1} \exp(-\| x_j - x_{n+1} \|)} - \frac{\exp(-\| x_i - x_{n+1} \|)}{\sum_{j=1}^{n+1} \exp(-\| x_j - x_i \|)} \right\} \\
        &\propto y' \left\{ \sum_{j=1}^{n+1} \Big[ \exp(-\| x_j - x_i \|) - \exp(- (\| x_j - x_{n+1} \| + \| x_i - x_{n+1} \|)) \Big] \right\}.
    \end{align*}    
    By the triangle inequality, $\exp(-\| x_j - x_i \|) \ge \exp( -(\| x_j - x_{n+1} \| + \| x_i - x_{n+1} \|))$, for all $i,j = 1, \dots, n+1$. The factor multiplied by $y'$ above is therefore non-negative, and hence $\hat{y}'_{n+1} - \hat{y}'_i$ is an increasing function of $y'$, for all $i$. The RDPS corresponding to this regression method can therefore be calculated efficiently using two runs of the algorithm. While this approach is not robust, we can simply trim the predictions by replacing $y_j$ and $y'$ above with $\max\{\min\{y_j, t_1\}, t_2\}$ and $\max\{\min\{y', t_1\}, t_2\}$, for some cut-off thresholds $t_1,t_2\in\R$.
\end{example}

Even if $\hat{y}'_{n+1} - \hat{y}'_i$ is not increasing in $y'$, it may still be possible to efficiently implement the RDPS. This is possible, for example, when the predictions are linear functions of the observations in the training data. In particular, if $\hat{y}_i = a_i y' + b_i$, where $a_i$ and $b_i$ are constants in the sense that they do not depend on $y'$ and only on $x_1, \dots, x_{n+1}$ and $y_1, \dots, y_n$, then we have
\[
G_{x_{n+1}}^{\Pnp(y')}(y) = \frac{1}{n+1} \left( \sum_{i=1}^{n} \one\{(a_{n+1} - a_i)y' \le y - b_i)\} + \one\{y' \le y\}\right).
\]
In this case, the indicator functions only change at the values $c_i(y) = (y - b_i)/(a_{n+1} - a_i)$, $i = 1, \dots, n$, and $c_{n+1}(y) = y$. Hence, the bounds $\Pi_{\ell,x_{n+1}}^{\Pn}(y)$ and $\Pi_{u,x_{n+1}}^{\Pn}(y)$ can be calculated by evaluating $G_{x_{n+1}}^{\Pnp(y')}(y)$ at $y' \le c_{(1)}(y)$, $y' \in (c_{(i-1)}(y), c_{(i)}(y)]$, for $i = 2, \dots, n+1$, and $y' > c_{(n+1)}(y)$, and taking the minimum and maximum over these $n+1$ distinct values.

\end{document}